\documentclass[copyright,creativecommons]{eptcs}
\providecommand{\event}{HOR 2010}
\usepackage{breakurl}
\usepackage[english]{babel}
\usepackage{color}
\usepackage{amssymb,amsmath,amsthm}
\usepackage{subfig,float}
\usepackage{rotating,xspace}
\usepackage{tikz}
\hypersetup{pdfborder = {0 0 0}}

\usetikzlibrary{positioning,fit,arrows,decorations.markings}
\tikzset{
to/.style={->,},
tob/.style={->,bend right=40},
tos/.style={->},
eq/.style={color=white},
}

\newcommand{\toAB}[2]{%
to node[yshift=0.4em,right=1.1] {\scriptsize $#1$}%
   node[yshift=-0.6em,right=1.3] {\scriptsize \makebox[0mm][r]{$#2$}}%
}
\newcommand{\tobAB}[2]{%
to node[yshift=0.9em,right=1.1] {\scriptsize $#1$}%
   node[yshift=0.5em,right=1.3] {\scriptsize \makebox[0mm][r]{$#2$}}%
}
\newcommand{\doAB}[2]{%
to node[xshift=0.4em,below=0.2] {\scriptsize $#1$}%
   node[xshift=-0.4em,below=0.2] {\scriptsize \makebox[0mm][r]{$#2$}}%
}
\theoremstyle{definition}
\newtheorem{definition}{Definition}
\newtheorem{example}[definition]{Example}
\theoremstyle{plain}
\newtheorem{lemma}[definition]{Lemma}
\newtheorem{theorem}[definition]{Theorem}

\DeclareMathAlphabet{\mathcal}{OMS}{xmdcmsy}{m}{n}
\DeclareSymbolFont{letters}{OML}{cmbboard}{m}{it}

\newlength{\tlen}
\let\QEDSYMBOL\qedsymbol

\newcommand\relgt{\succ}

\newcommand{\HOLE}{\ensuremath{\Box}\xspace}
\newcommand{\QED}{\hfill\QEDSYMBOL}

\newcommand{\UDE}[1]{\UU^{\scalebox{0.6}{+}}_\eta\vspace{-0.2ex}{(#1)}}

\newcommand\BM{\begin{pmatrix}}
\newcommand\EM{\end{pmatrix}}
\newcommand{\NE}{\hspace{-0.4em}&\hspace{-0.4em}}
\newcommand{\NR}{\\}
\newcommand{\NF}{\mi{NF}}
\newcommand{\TTT}{\textsf{T\!\raisebox{-1mm}{T}\!T}\xspace}
\newcommand{\TTTT}{$\TTT\textsf{\!\raisebox{-1mm}{2}}$\xspace}
\newcommand{\m}[1]{\mathsf{#1}}
\newcommand{\mi}[1]{\mathit{#1}}
\renewcommand{\AA}{\mathcal{A}}
\newcommand{\CC}{\mathcal{C}}
\newcommand{\FF}{\mathcal{F}}
\newcommand{\MM}{\mathcal{M}}
\newcommand{\NN}{\mathbb{N}}
\newcommand{\OO}{\mathcal{O}}
\newcommand{\PP}{\mathcal{P}}
\newcommand{\RR}{\mathcal{R}}
\renewcommand{\SS}{\mathcal{S}}
\newcommand{\UU}{\mathcal{U}}
\newcommand{\VV}{\mathcal{V}}

\newcommand{\Pos}{\PP\m{os}}
\newcommand{\aar}{\mathrm{aa}}
\newcommand\TERMS[2]{\mathcal{T}(#1,#2)}
\newcommand\FVTERMS{\TERMS{\FF}{\VV}}

\newcommand{\CURRYf}[1]{{{#1}{\downarrow}_{\CC(\FF)}}}
\newcommand{\CURRY}[1]{{{#1}{\downarrow}_\CC}}
\newcommand{\CURRYp}[1]{{{#1}{\downarrow}_{\CC'}}}
\newcommand{\UNCURRYr}[1]{{{#1}{\downarrow}_{\UU(\RR)}}}
\newcommand{\UNCURRY}[1]{{{#1}{\downarrow}_\UU}}

\newcommand{\superterm}{\mathrel{{\trianglerighteq}}}
\newcommand{\prsuperterm}{\mathrel{{\vartriangleright}}}
\newcommand{\seq}[2][n]{{#2_1},\dots,{#2_{#1}}}
\newcommand{\ieto}{\mathrel{\smash{\overset{\raisebox{-2pt}{\scriptsize
$\m{i}$}}{\to}}}^\epsilon}
\newcommand{\ito}{\mathrel{\smash{\overset{\raisebox{-2pt}{\scriptsize
$\m{i}$}}{\to}}}}

\newcommand{\rito}{\mathrel{\smash{\overset{\raisebox{-2pt}{\scriptsize
$\m{ri}$}}{\to}}}}

\newcommand{\jc}{{~}}
\newcommand{\dl}[2]{\mathrm{dh}({#2},{#1})}
\newcommand{\dc}[2]{\mathrm{dc}_{#1}({#2})}
\newcommand{\idc}[2]{\mathrm{idc}_{#1}({#2})}
\newcommand{\Nat}{\ensuremath{\mathbb{N}}\xspace}
\newcommand{\from}{\leftarrow}
\def\test#1#2#3{\setbox0=\hbox{$\vphantom{#1}^{#2}_{#3}$}%
                \dimen0=\wd0%
                \setbox1=\hbox{$\scriptstyle #2$}%
                \advance\dimen0-\wd1%
                \setbox1=\hbox{\hskip\dimen0\copy1}%
                \dimen0=\wd0%
                \setbox2=\hbox{$\scriptstyle #3$}%
                \advance\dimen0-\wd2%
                \setbox2=\hbox{\hskip\dimen0\copy2}%
                {\vphantom{#1}^{\box1}_{\box2}}{#1}
}
\newcommand{\fromTB}[2]{\mathrel{\test{\from}{#1}{#2}}}
\newcommand{\MINISMT}{\textsf{MiniSmt}\xspace}

\title{Uncurrying for Innermost Termination and Derivational Complexity%
\thanks{%
This research is supported by FWF (Austrian Science Fund) project P18763
and the Grant-in-Aid for Young Scientists Nos.\,20800022 and 22700009 of 
the Japan Society for the Promotion of Science.
}
}
\author{%
Harald Zankl,\textsuperscript{1}
Nao Hirokawa,\textsuperscript{2} 
and Aart Middeldorp\textsuperscript{1}%
\institute{%
\textsuperscript{1}
Institute of Computer Science,
University of Innsbruck, Austria 
\email{$\{$harald.zankl,aart.middeldorp$\}$@uibk.ac.at}
}
\institute{%
\textsuperscript{2}
School of Information Science,
Japan Advanced Institute of Science and Technology, Japan 
\email{hirokawa@jaist.ac.jp}
}
}
\def\titlerunning{Uncurrying}
\def\authorrunning{H.~Zankl, N.~Hirokawa \& A.~Middeldorp}

\begin{document}
\maketitle

\begin{abstract}
First-order applicative term rewriting systems provide a natural
framework for modeling higher-order aspects. In 
earlier work we introduced an uncurrying transformation
which is termination
preserving and reflecting. In this paper we investigate how this
transformation behaves for innermost termination and (innermost)
derivational complexity. We prove that it reflects innermost termination
and innermost derivational complexity and that it preserves and
reflects polynomial derivational complexity. For the preservation of
innermost termination and innermost derivational complexity we give
counterexamples. Hence uncurrying may be used as a preprocessing
transformation for innermost termination proofs and establishing
polynomial upper and lower bounds on the derivational complexity.
Additionally it may be used to establish upper bounds on the innermost
derivational complexity while it neither is sound for proving innermost
non-termination nor for obtaining lower bounds on the innermost
derivational complexity.
\end{abstract}

\section{Introduction}

Proving termination of first-order applicative term rewrite systems is
challenging since the rules lack sufficient structure. But these systems 
are important since they provide a natural framework for modeling
higher-order aspects found in functional programming languages.
Since proving termination is easier for
innermost than for full rewriting we lift some of
the recent results from \cite{HMZ08} from full to innermost termination.
For the properties that do not transfer to the innermost setting we
provide counterexamples. Furthermore we show that the uncurrying
transformation is suitable for proving upper bounds on the (innermost)
derivational complexity.

We remark that our approach on proving innermost termination also is
beneficial for functional programming languages that adopt a lazy 
evaluation strategy since applicative term rewrite systems modeling
functional programs are left-linear and non-overlapping. It is well 
known that for this class of systems termination and innermost termination
coincide (see~\cite{G95} for a more general result).

\smallskip 

The remainder of this paper is organized as follows.
After recalling preliminaries in
Section~\ref{PRE:main}, we show that uncurrying preserves innermost
non-termination (but not innermost termination) in Section~\ref{INN:main}.
In Section~\ref{DER:main} we show that it preserves and reflects
derivational complexity of rewrite systems while it only reflects
innermost derivational complexity.
Section~\ref{EXP:main} reports on experimental results
and we conclude in Section~\ref{conclusion}.

\section{Preliminaries}
\label{PRE:main}

In this section we fix preliminaries on rewriting, complexity and
uncurrying.

\subsection{Term Rewriting}

\newcommand{\Emph}[1]{#1}
We assume familiarity with term rewriting \cite{BN98,TeReSe}.
Let $\FF$ be a signature and $\VV$ a set of variables disjoint from~$\FF$.
By $\FVTERMS$ we denote the set of terms over $\FF$ and $\VV$.
The \Emph{size} of a term $t$ is denoted $|t|$.
A \Emph{rewrite rule} is a pair of terms $(\ell,r)$, written $\ell \to r$,
such that $\ell$ is not a variable and all variables in $r$
occur in $\ell$.
A \Emph{term rewrite system} (TRS for short) is a set of rewrite rules.
A TRS $\RR$ is said to be \Emph{duplicating} if there exist a rewrite
rule $\ell \to r \in \RR$ and a variable $x$ that occurs more often in $r$ 
than in $\ell$. 

\smallskip 

\Emph{Contexts} are terms over the
signature $\FF \cup \{\HOLE\}$ with exactly one occurrence of the fresh
constant \HOLE (called \Emph{hole}). The expression $C[t]$
denotes the result of replacing the hole in $C$ by the term $t$.
A \Emph{substitution} $\sigma$ is a mapping from variables to terms
and $t\sigma$ denotes the result of replacing the variables in $t$
according to $\sigma$. Substitutions may change only finitely many
variables (and are thus written as
$\{ x_1 \mapsto t_1, \dots, x_n \mapsto t_n \}$).
The \Emph{set of positions} of a term $t$ is defined as
$\Pos(t) = \{ \epsilon \}$ if $t$ is a variable and as
$\Pos(t) = \{ \epsilon \} \cup \{ iq \mid q \in \Pos(t_i) \}$
if $t = f(\seq{t})$. 
Positions are used to address occurrences of subterms.
The \Emph{subterm of $t$ at position $p \in \Pos(t)$} is defined as
$t|_p = t$ if $p = \epsilon$ and as
$t|_p = t_i|_q$ if $p = iq$. We say a position $p$ is
\Emph{to the right of} a position $q$ if $p = p_1ip_2$ and $q = q_1jq_2$
with $p_1 = q_1$ and $i > j$. For a term $t$ and positions
$p, q \in \Pos(t)$ we say $t|_{p}$ is to the right of $t|_{q}$ if $p$ is
to the right of $q$.

\smallskip 

A \Emph{rewrite relation} is a binary relation on terms that is
closed under contexts and substitutions. For a TRS $\RR$ we define
$\to_\RR$ to be the smallest rewrite relation that contains $\RR$.
We call $s \to_\RR t$ a \Emph{rewrite step} if there exist a
context $C$, a rewrite rule $\ell \to r \in \RR$, and a
substitution $\sigma$ such that $s = C[\ell\sigma]$ and $t = C[r\sigma]$. 
In this case we call $\ell\sigma$ a \Emph{redex} and say that
$\ell\sigma$ has been \Emph{contracted}.
A \Emph{root rewrite step}, denoted by $s \to_\RR^\epsilon t$,
has the shape $s = \ell\sigma \to_\RR r\sigma = t$
for some $\ell \to r \in \RR$.
A \Emph{rewrite sequence} is a sequence of rewrite steps.
The \Emph{set of normal forms} of a TRS $\RR$ is defined as
$\NF(\RR) = \{ t \in \FVTERMS \mid \text{$t$ contains no redexes} \}$.
A redex $\ell\sigma$ in a term $t$ is called \Emph{innermost}
if proper subterms of $\ell\sigma$ are normal forms,
and \Emph{rightmost innermost} if in addition $\ell\sigma$ is to the right
of any other redex in $t$.
A rewrite step is called \Emph{innermost} (\Emph{rightmost innermost})
if an innermost (rightmost innermost) redex is contracted, written
$\ito$ and $\rito$, respectively.

\smallskip 

If the TRS $\RR$ is not essential or clear from the context the subscript
${}_\RR$ is omitted in $\to_\RR$ and its derivatives.
As usual, $\to^+$ ($\to^*$) denotes the transitive (reflexive and
transitive) closure of $\to$ and $\to^m$ its $m$-th iterate.
A TRS is \Emph{terminating} (\Emph{innermost terminating}) if $\to^+$
($\ito^+$) is well-founded.

\smallskip 

Let $\PP$ be a property of TRSs and let $\Phi$ be a transformation on
TRSs with $\Phi(\RR) = \RR'$. We say $\Phi$ \emph{preserves} 
$\PP$ if $\PP(\RR)$ implies $\PP(\RR')$ and $\Phi$ \emph{reflects} 
$\PP$ if $\PP(\RR')$  implies $\PP(\RR)$. Sometimes we call $\Phi$
$\PP$ preserving if $\Phi$ preserves $\PP$ and $\PP$ reflecting if $\Phi$
reflects $\PP$, respectively.

\subsection{Derivational Complexity}

For complexity analysis we assume TRSs to be finite and (innermost)
terminating.

\smallskip 

Hofbauer and Lautemann~\cite{HL89} introduced the concept of
derivational complexity for terminating TRSs. The idea is to measure
the maximal length of rewrite sequences (derivations) depending on the
size of the starting term. Formally, the \emph{derivation height} of a
term $t$ (with respect to a finitely branching and well-founded
order $\to$) is defined on natural numbers as 
$\dl{\to}{t} = \max \{ m \in \NN \mid \text{$t \to^m u$ for some $u$} \}$.
The \emph{derivational complexity} $\dc{\RR}{n}$ of a TRS $\RR$ is 
then defined as
$\dc{\RR}{n} = \max \{ \dl{\to_\RR}{t} \mid |t| \leqslant n \}$.
Similarly we define
the \emph{innermost} derivational complexity as
$\idc{\RR}{n} = \max \{ \dl{\ito_\RR}{t} \mid |t| \leqslant n \}$.
Since we regard finite TRSs only, these functions are well-defined
if $\RR$ is (innermost) terminating. If $\dc{\RR}{n}$ is bounded
by a linear, quadratic, cubic, $\dots$ function or polynomial, $\RR$
is said to have linear, quadratic, cubic, $\dots$ or polynomial
derivational complexity. A similar convention applies to $\idc{\RR}{n}$.

\smallskip 

For functions $f,g \colon \Nat \to \Nat$ we write $f(n) \in \OO(g(n))$ if
there are constants $M, N \in \Nat$ such that
$f(n) \leqslant M \cdot g(n) + N$ for all $n \in \Nat$. 

\smallskip 

One popular method to prove polynomial upper bounds on the derivational
complexity is via triangular matrix interpretations \cite{MSW08}, which
are a special instance of monotone algebras.
An \Emph{$\FF$-algebra} $\AA$ consists of a non-empty carrier $A$
and a set of interpretations $f_\AA$ for every $f \in \FF$. By
$[\alpha]_\AA(\cdot)$
we denote the usual evaluation function
of $\AA$ according to an assignment $\alpha$
which maps variables to values in~$A$.
An $\FF$-algebra $\AA$ together with a well-founded order $\relgt$
on $A$ is called a \emph{monotone algebra} if
every $f_\AA$ is monotone with respect to $\relgt$.
Any monotone algebra $(\AA,{\relgt})$ induces a well-founded
order on terms: $s \relgt_\AA t$ if for any assignment $\alpha$
the condition $[\alpha]_\AA(s) \relgt [\alpha]_\AA(t)$ holds.
A TRS $\RR$ is \Emph{compatible} with a monotone
algebra $(\AA,{\relgt_\AA})$ if $l \relgt_\AA r$ for all $l \to r \in \RR$.

\smallskip 

\emph{Matrix interpretations} $(\MM,{\relgt})$
(often just denoted $\MM$) are a special form of monotone algebras.
Here the carrier is $\Nat^{d}$ for some fixed
dimension $d \in \Nat \setminus \{ 0 \}$. The order
$\relgt$ is defined on $\NN^d$ as
$(\seq[d]{u}) \relgt (\seq[d]{v})$ if 
$u_1 >_{\Nat} v_1$ and
$u_i \geqslant_{\Nat} v_i$ for all $2 \leqslant i \leqslant d$.
If every $f \in \FF$ of arity $n$ is interpreted as
$f_\MM(\vec{x_1},\dots,\vec{x_n}) = F_1\vec{x_1} + \dots + F_n\vec{x_n} +
\vec{f}$ where $F_i \in \Nat^{d\times d}$ 
for all $1 \leqslant i \leqslant n$
and $\vec{f} \in \Nat^d$ then
monotonicity of $\relgt$ is achieved by demanding
${F_i}_{(1,1)} \geqslant 1$ for any $1 \leqslant i \leqslant n$.
Such interpretations have been introduced in \cite{EWZ08}.

\smallskip 

A matrix interpretation where for every $f \in \FF$ all $F_i$
($1 \leqslant i \leqslant n$ where $n$ is the arity of $f$)
are upper triangular is called \emph{triangular} (abbreviated by TMI).
A square matrix $A$ of dimension $d$ is of \emph{upper triangular}
shape if $A_{(i,i)} \leqslant 1$ and $A_{(i,j)} = 0$ if $i > j$ for all
$1 \leqslant i, j \leqslant d$.
The next theorem is from \cite{MSW08}.

\smallskip

\begin{theorem}
\label{THM:tmi}
If a TRS $\RR$ is compatible with a TMI $\MM$ of dimension $d$
then $\dc{\RR}{n} \in \OO(n^d)$.
\end{theorem}

\smallskip

Recent generalizations of this theorem are reported
in~\cite{NZM10,W10}.

\subsection{Uncurrying}

This section recalls definitions and results from \cite{HMZ08}.

\smallskip 

An \emph{applicative} term rewrite system (ATRS for short) is a TRS over
a signature that consists of constants and a single binary function
symbol called application 
which is
denoted by the infix and left-associative
symbol $\circ$. In examples we often use juxtaposition instead of $\circ$.
Every ordinary TRS can be transformed into an ATRS by currying.
Let $\FF$ be a signature. The currying system $\CC(\FF)$
consists of the rewrite rules
\[
f_{i+1}(x_1,\dots,x_i,y) \to f_i(x_1,\dots,x_i) \circ y 
\]
for every $n$-ary function symbol $f \in \FF$ and every
$0 \leqslant i < n$. Here $f_n = f$ and, for every $0 \leqslant i < n$,
$f_i$ is a fresh function symbol of arity $i$.
The currying system $\CC(\FF)$ is confluent and terminating. Hence
every term $t$ has a unique normal form $\CURRYf{t}$.
For instance,
$\m{f(a,b)}$ is transformed into $\m{f} \jc \m{a} \jc \m{b}$.
Note that we write $f$ for $f_0$.

\smallskip 

Next we recall the uncurrying transformation from \cite{HMZ08}.
Let $\RR$ be an ATRS over a signature $\FF$. The \emph{applicative arity}
$\aar(f)$ of a constant $f \in \FF$ is defined
as the maximum $n$ such that $f \circ t_1 \circ \cdots \circ t_n$ is a
subterm in the left- or right-hand side of a rule in $\RR$. This
notion is extended to terms as follows:
$\aar(t) = \aar(f)$ if $t$ is a constant $f$ and
$\aar(t_1) - 1$ if $t = t_1 \circ t_2$.
Note that $\aar(t)$ is undefined if the head symbol of $t$ is a
variable. The uncurrying system 
$\UU(\RR)$ consists of the rewrite rules
\[
f_i(x_1,\dots,x_i) \circ y \to f_{i+1}(x_1,\dots,x_i,y)
\]
for every constant $f \in \FF$ and every $0 \leqslant i < \aar(f)$. 
Here $f_0 = f$ and, for every $i > 0$, $f_i$ is a fresh function symbol of
arity $i$. We say that $\RR$ is \emph{left head variable free} if
$\aar(t)$ is defined for every non-variable subterm $t$ of a left-hand
side of a rule in $\RR$.
This means that no subterm of a left-hand side in $\RR$ is
of the form $t_1 \circ t_2$ where $t_1$ is a variable.
The uncurrying system $\UU(\RR)$, or simply $\UU$, is confluent and
terminating. Hence every term $t$ has a unique normal form $\UNCURRY{t}$.
The \emph{uncurried} system $\UNCURRY{\RR}$ is the TRS consisting of the
rules $\UNCURRY{\ell} \to \UNCURRY{r}$ for every $\ell \to r \in \RR$.
However the rules of $\UNCURRY{\RR}$ are not enough to simulate an
arbitrary rewrite sequence in $\RR$. The natural idea is now to add
$\UU(\RR)$, but still $\UNCURRYr{\RR} \cup \UU(\RR)$ is not enough
as shown in the next example from \cite{HMZ08}.

\begin{table}
\begin{xalignat*}{5}
&\makebox[\tlen][c]{$\RR$} &
&\makebox[\tlen][c]{$\UU(\RR)$} &
&\makebox[\tlen][c]{$\UNCURRYr{\RR}$} &
&\makebox[\tlen][c]{$\RR_\eta$} &
&\makebox[\tlen][c]{$\UNCURRYr{\RR_\eta}$}
\\
\hline
\m{id} \jc x   &\to x &
\m{id} \circ x &\to \m{id}_1(x) &
\m{id}_1(x)    &\to x &
\m{id} \jc x   &\to x &
\m{id}_1(x)    &\to x
\\
\m{f} \jc x         &\to \m{id} \jc \m{f} \jc x &
\m{id}_1(x) \circ y &\to \m{id}_2(x,y) &
\m{f}_1(x)          &\to \m{id}_2(\m{f},x) &
\m{f} \jc x         &\to \m{id} \jc \m{f} \jc x &
\m{f}_1(x)          &\to \m{id}_2(\m{f},x)
\\
&&
\m{f} \circ x &\to \m{f}_1(x) &
&&
\m{id} \jc x \jc y  &\to x \jc y &
\m{id}_2(x,y)       &\to x \circ y
\\
\hline
\end{xalignat*}
\vspace{-3em}
\caption{Some (transformed) TRSs}
\label{TAB:trafo}
\end{table}

\begin{example}
Consider the TRS~$\RR$ in Table~\ref{TAB:trafo}. 
Based on $\aar(\m{id}) = 2$ and $\aar(\m{f}) = 1$ we get three rules
in $\UU(\RR)$ and can compute $\UNCURRYr{\RR}$.
The TRS~$\RR$ is
non-terminating but $\UNCURRYr{\RR} \cup \UU(\RR)$ is terminating.
\end{example}

\smallskip

Let $\RR$ be a left head variable free ATRS. The
\emph{$\eta$-saturated} ATRS $\RR_\eta$ is the smallest extension of
$\RR$ such that $\ell \circ x \to r \circ x \in \RR_\eta$ whenever
$\ell \to r \in \RR_\eta$ and $\aar(\ell) > 0$. Here $x$ is a variable
that does not appear in $\ell \to r$. In the following we write
$\UDE{\RR}$ for $\UNCURRYr{\RR_\eta} \cup \UU(\RR)$. 
Note that applicative arities are computed before $\eta$-saturation.

\begin{example}
Consider again Table~\ref{TAB:trafo}. Since $\aar(\m{id}) = 2$ but 
$\aar(\m{id} \jc x) = 1$ for the rule $\m{id} \jc x \to x$ in $\RR$
this explains the rule $\m{id} \jc x \jc y \to x \jc y$ in $\RR_\eta$.
Note that $\UDE{\RR}$ is non-terminating.
\end{example}

\smallskip

For a term $t$ over the signature of the TRS $\UDE{\RR}$, 
we denote by
$\CURRYp{t}$ the result of identifying different function symbols in
$\CURRY{t}$ that originate from the same function symbol in $\FF$. 
For a substitution $\sigma$, we write $\UNCURRY{\sigma}$ for the
substitution $\{ x \mapsto \UNCURRY{\sigma(x)} \mid \text{$x \in \VV$} \}$.

\begin{center}
\textbf{\textit{From now on we assume that every ATRS is left-head
variable free.}}
\end{center}

We conclude this preliminary section by recalling some results
from \cite{HMZ08}.

\smallskip

\begin{lemma}[\textnormal{\cite[Lemma~20]{HMZ08}}]
\label{uncurry substitution}
Let $\sigma$ be a substitution. If $t$ is head
variable free then $\UNCURRY{t}\UNCURRY{\sigma} = \UNCURRY{(t\sigma)}$.
\qed
\end{lemma}

\smallskip

\begin{lemma}[\textnormal{\cite[Lemma~15]{HMZ08}}]
\label{R vs Reta}
If $\RR$ is an ATRS then ${\to_\RR} = {\to_{\RR_\eta}}$.
\qed
\end{lemma}

\smallskip

\begin{lemma}[\textnormal{\cite[Lemmata~26 and~27]{HMZ08}}]
\label{curried step}
Let $\RR$ be an ATRS. If $s$ and $t$ are terms over the signature of
$\UDE{\RR}$
then (1) $s \to_{\UNCURRY{\RR}} t$ if and only if
$\CURRYp{s} \to_{\RR} \CURRYp{t}$ and (2) $s \to_\UU t$ implies
$\CURRYp{s} = \CURRYp{t}$.
\qed
\end{lemma}

\smallskip

\begin{lemma}[\textnormal{\cite[Proof of Theorem~16]{HMZ08}}]
\label{uncurried step}
Let $\RR$ be an ATRS. If $s \to_{\RR} t$ then
$\UNCURRY{s} \to^+_{\UDE{\RR}} \UNCURRY{t}$.
\qed
\end{lemma}

Consequently our transformation is shown to be termination preserving and 
reflecting.

\begin{theorem}[\textnormal{\cite[Theorems~16 and~28]{HMZ08}}]
\label{main}
Let $\RR$ be an ATRS. 
The ATRS $\RR$ is terminating if and only if the TRS $\UDE{\RR}$ is
terminating.
\qed
\end{theorem}

\section{Innermost Uncurrying}
\label{INN:main}

Before showing that our transformation reflects innermost termination
we show that it does not preserve innermost termination. 
Hence uncurrying may not be used as a preprocessing transformation for 
innermost non-termination proofs.

\begin{example}
\label{EX:reflecti}
Consider the ATRS
$\RR$ consisting of the rules
\begin{xalignat*}{2}
\m{f} \jc x &\to \m{f} \jc x & \m{f} &\to \m{g}
\end{xalignat*}
In an innermost sequence the first rule is never applied and hence
$\RR$ is innermost terminating.
The TRS $\UDE{\RR}$ consists of the rules
\begin{xalignat*}{4}
\m{f}_1(x) &\to \m{f}_1(x) &
\m{f} &\to \m{g} &
\m{f}_1(x) &\to \m{g} \circ x &
\m{f} \circ x &\to \m{f}_1(x)
\end{xalignat*}
and is not innermost terminating due to the rule
$\m{f}_1(x) \to \m{f}_1(x)$.
\end{example}

\smallskip

The next example shows that $s \ito_\RR t$ does not imply
$s{\downarrow_\UU} \ito_{\UDE{\RR}}^+ t{\downarrow}_\UU$.
This is not a counterexample to soundness of uncurrying for innermost
termination, but it shows that the proof for the ``if-direction'' of
Theorem~\ref{main} (which is based on Lemma~\ref{uncurried step}) cannot
be adopted for the innermost case without further ado.

\smallskip

\begin{example}
\label{simulation - counterexample}
Consider the ATRS $\RR$ consisting of the rules
\begin{xalignat*}{3}
\m{f} &\to \m{g} &
\m{a} &\to \m{b} &
\m{g} \jc x &\to \m{h}
\end{xalignat*}
and the innermost step $s = \m{f} \jc \m{a} \ito_\RR \m{g} \jc \m{a} = t$.
We have $s{\downarrow_\UU} = \m{f} \circ \m{a}$ and
$t{\downarrow}_\UU = \m{g}_1(\m{a})$.
The TRS $\UDE{\RR}$ consists of the rules
\begin{xalignat*}{4}
\m{f} &\to \m{g} &
\m{a} &\to \m{b} &
\m{g}_1(x) &\to \m{h} &
\m{g} \circ x &\to \m{g}_1(x)
\end{xalignat*}
We have
$s{\downarrow_\UU} \ito_{\UDE{\RR}} \m{g} \circ \m{a}$ but the
step from $\m{g} \circ \m{a}$ to $t{\downarrow}_\UU$ is not innermost.
\end{example}

\smallskip

The above problems can be solved if we consider terms that are not
completely uncurried. The 
next lemmata
prepare for the proof.
Below we write $s \prsuperterm t$ if $t$ is a proper subterm of $s$.

\begin{lemma}
\label{LEM:nfb}
Let $\RR$ be an ATRS. If $s$ is a term over the signature of $\RR$,
$s \in \NF(\RR)$, and $s \to_\UU^* t$ then
$t \in \NF(\UNCURRY{\RR_\eta})$.
\end{lemma}
\proof
From Lemma~\ref{curried step}(2) we obtain $\CURRYp{s} = \CURRYp{t}$. Note
that $\CURRYp{s} = s$ because $s$ is a term over the signature of $\RR$.
If $t \notin \NF(\UNCURRY{\RR_\eta})$ then $t \to_{\UNCURRY{\RR_\eta}} u$
for some term $u$. Lemma~\ref{curried step}(1) yields
$\CURRYp{t} \to_{\RR_\eta} \CURRYp{u}$ and Lemma~\ref{R vs Reta} yields
$s \to_\RR \CURRYp{u}$. Hence $s \notin \NF(\RR)$, contradicting the
assumption.
The proof is summarized in the following diagram:
\[
\begin{tikzpicture}
\node (u1)                   {$s$};
\node (u2) [right=of u1,xshift=5em]     {$t$};
\node (u3) [right=of u2,xshift=5em]     {$u$};
\node (d1) [below=of u1]     {$\CURRYp{s}$};
\node (d2) [below=of u2]     {$\CURRYp{t}$};
\node (d3) [below=of u3]     {$\CURRYp{u}$};
\node (dx) [below=of u3,yshift=-1em] {};
\draw[tos] (u1) \toAB{*}{\UU} (u2);
\draw[to]  (u2) \toAB{}{\UNCURRY{\RR_\eta}} (u3);
\draw[eq]  (d1) to node[above,color=black]
 {\scriptsize Lemma~\ref{curried step}(2)} 
                   node[color=black] {=} (d2);
\draw[to]  (d2) \toAB{}{\RR_\eta} node[above,color=black]
 {\scriptsize Lemma~\ref{curried step}(1)} (d3);
\draw[tob] (d2) \tobAB{}{\RR} node[below]
 {\scriptsize Lemma~\ref{R vs Reta}} (dx);
\draw[eq]  (u1)to node[color=black] {=} (d1);
\draw[tos] (u2) \doAB{*}{\CC'} (d2);
\draw[tos] (u3) \doAB{*}{\CC'} (d3);
\end{tikzpicture}
\]
\qed

\begin{lemma}
\label{LEM:commute0}
$\to^*_\UU \cdot \prsuperterm {\subseteq} \prsuperterm \cdot \to^*_\UU$
\end{lemma}
\proof
Assume $s \to^*_\UU t \prsuperterm u$. We show that
$s \prsuperterm \cdot \to^*_\UU u$ by induction on $s$. If $s$ is a
variable or a constant then there is nothing to show. So let
$s = s_1 \circ s_2$. We consider two cases.
\begin{itemize}
\item
If the outermost $\circ$ has not been uncurried then $t = t_1 \circ t_2$
with $s_1 \to^*_\UU t_1$ and $s_2 \to^*_\UU t_2$. Without loss of
generality assume that $t_1 \superterm u$. If $t_1 = u$ then
$s \prsuperterm s_1 \to^*_\UU t_1$. If $t_1 \prsuperterm u$ then the
induction hypothesis yields $s_1 \prsuperterm \cdot \to^*_\UU u$ and
hence also $s \prsuperterm \cdot \to^*_\UU u$.
\item
If the outermost $\circ$ has been uncurried in the sequence from $s$ to
$t$ then the head symbol of $s_1$ cannot be a variable and
$\aar(s_1) > 0$. Hence we may write
$s_1 = f \circ t_1 \circ \dots \circ t_i$ and 
$t = f_{i+1}(t_1',\dots,t_i',s_2')$ with $t_j \to^*_\UU t_j'$ for all
$1 \leqslant j \leqslant i$ and $s_2 \to^*_\UU  s_2'$. Clearly,
$t_j' \superterm u$ for some $1 \leqslant j \leqslant i$ or
$s_2' \superterm t$. In all cases the result follows with the same
reasoning as in the first case.
\qed
\end{itemize}

\smallskip

The next lemma states (a slightly more general result than) that an
innermost root rewrite step in an ATRS $\RR$ can be simulated by an
innermost rewrite sequence in $\UDE{\RR}$.

\smallskip

\begin{lemma}
\label{LEM:ito}
For every ATRS $\RR$ the inclusion
${\fromTB{*}{\UU} \cdot \ieto_\RR} \subseteq
{\ito^+_{\UDE{\RR}} \cdot \fromTB{*}{\UU}}$ holds.
\end{lemma}
\begin{proof}
We prove that $s \ito^+_{\UDE{\RR}} \UNCURRY{r}\UNCURRY{\sigma}
\fromTB{*}{\UU} r\sigma$ whenever
$s \fromTB{*}{\UU} \ell\sigma \ieto_\RR r\sigma$ for some rewrite rule
$\ell \to r$ in $\RR$. By Lemma~\ref{uncurry substitution} and the
confluence of $\UU$,
\[
s \ito_\UU^* \UNCURRY{(\ell\sigma)} = \UNCURRY{\ell}\UNCURRY{\sigma}
\to_{\UDE{\RR}} \UNCURRY{r}\UNCURRY{\sigma} \fromTB{*}{\UU} r\sigma
\]
It remains to show that the sequence $s \ito_\UU^* \UNCURRY{(\ell\sigma)}$
and the step $\UNCURRY{\ell}\UNCURRY{\sigma} \to_{\UDE{\RR}}
\UNCURRY{r}\UNCURRY{\sigma}$ are innermost with respect to
$\UDE{\RR}$.
For the former, let
$s \ito_\UU^* C[u] \ito_\UU C[u'] \ito_\UU^* \UNCURRY{(\ell\sigma)}$
with $u \ieto_\UU u'$ and let $t$ be a proper subterm of $u$. Obviously
$\ell\sigma \to^*_\UU C[u] \prsuperterm t$. According to
Lemma~\ref{LEM:commute0}, $\ell\sigma \prsuperterm v \to^*_\UU t$ for some
term $v$. Since $\ell\sigma \ito_\RR^\epsilon r\sigma$,
the term $v$ is a normal
form of $\RR$. Hence $t \in \NF(\UNCURRY{\RR_\eta})$ by
Lemma~\ref{LEM:nfb}. Since $u \ieto_\UU u'$, $t$ is also a normal form of
$\UU$. Hence $t \in \NF(\UDE{\RR})$ as desired.
For the latter, let $t$ be a proper subterm of $\UNCURRY{(\ell\sigma)}$.
According to Lemma~\ref{LEM:commute0},
$\ell\sigma \prsuperterm u \to^*_\UU t$. The term $u$ is a normal form of
$\RR$. Hence $t \in \NF(\UNCURRY{\RR_\eta})$ by Lemma~\ref{LEM:nfb}.
Obviously, $t \in \NF(\UU)$ and thus also $t \in \NF(\UDE{\RR})$.
\end{proof}

The next example shows that 
it is not sound to replace
$\ito^\epsilon_\RR$ by $\ito_\RR$
in Lemma~\ref{LEM:ito}.

\smallskip

\begin{example}
Consider the ATRS $\RR$ consisting of the rules
\begin{xalignat*}{3}
\m{f} &\to \m{g} &
\m{f} \jc x &\to \m{g} \jc x &
\m{a} &\to \m{b}
\end{xalignat*}
Consequently the TRS $\UDE{\RR}$ consists of the rules
\begin{xalignat*}{5}
\m{f} &\to \m{g} &
\m{f}_1(x) &\to \m{g}_1(x) &
\m{a} &\to \m{b}&
\m{f} \circ x &\to \m{f}_1(x)&
\m{g} \circ x &\to \m{g}_1(x)
\end{xalignat*}
We have
$\m{f}_1(\m{a}) \fromTB{*}{\UU} \m{f} \circ \m{a} \ito_\RR
\m{g} \circ \m{a}$ but
$\m{f}_1(\m{a}) \ito^+_{\UDE{\RR}} \cdot \fromTB{*}{\UU}
\m{g} \circ \m{a}$ does not hold. To see that the latter does not hold,
consider the two reducts of $\m{g} \circ \m{a}$ with respect to
$\to_\UU^*$: $\m{g}_1(\m{a})$ and $\m{g} \circ \m{a}$. We have
neither
$\m{f}_1(\m{a}) \ito^+_{\UDE{\RR}} \m{g}_1(\m{a})$ nor
$\m{f}_1(\m{a}) \ito^+_{\UDE{\RR}} \m{g} \circ \m{a}$.
\end{example}

In order to extend Lemma~\ref{LEM:ito} to non-root positions, we
have to use rightmost innermost evaluation.
This avoids the situation in the above example
where parallel redexes become nested by uncurrying.

\smallskip

\begin{lemma}
\label{LEM:commute}
For every ATRS $\RR$ the inclusion
$\fromTB{*}{\UU} \cdot \rito_\RR {\subseteq} \ito^+_{\UDE{\RR}} \cdot
\fromTB{*}{\UU}$ holds.
\end{lemma}
\begin{proof}
\renewcommand{\qedsymbol}{}
Let $s \fromTB{*}{\UU} t = C[\ell\sigma] \rito_\RR C[r\sigma] = u$ with
$\ell\sigma \ieto_\RR r\sigma$. We use induction on $C$. If $C = \Box$
then $s \fromTB{*}{\UU} t \ieto_\RR u$. Lemma~\ref{LEM:ito} yields
$s \ito^+_{\UDE{\RR}} \cdot \fromTB{*}{\UU} u$. For the induction
step we consider two cases.
\begin{itemize}
\item
Suppose $C = \Box \circ s_1 \circ \dots \circ s_n$ and $n > 0$.
Since $\RR$ is left head variable free, $\aar(\ell)$ is defined. If
$\aar(\ell) = 0$ then
\(
s = t' \circ s_1' \circ \dots \circ s_n'
\fromTB{*}{\UU} \ell\sigma \circ s_1 \circ \dots \circ s_n
\ito_\RR r\sigma \circ s_1 \circ \dots \circ s_n
\)
with $t' \fromTB{*}{\UU} \ell\sigma$ and
$s_j' \fromTB{*}{\UU} s_j$ for $1 \leqslant j \leqslant n$.
The claim follows using Lemma~\ref{LEM:ito} and the fact that innermost
rewriting is closed under contexts.
If $\aar(\ell) > 0$ we
have to consider two cases. In the case where the leftmost $\circ$ 
symbol in $C$ has not been uncurried we proceed as when $\aar(\ell) = 0$.
If the leftmost $\circ$ symbol of $C$ has been uncurried, we
reason as follows. We may write
$\ell\sigma = f \circ u_1 \circ \dots \circ u_k$ where $k < \aar(f)$. We
have
$t = f \circ u_1 \circ \dots \circ u_k \circ s_1 \circ \dots \circ s_n$
and $u = r\sigma \circ s_1 \circ \dots \circ s_n$.
There exists an $i$ with $1 \leqslant i \leqslant \min \{ \aar(f), k+n \}$
such that 
\(
s = f_i(u_1',\dots,u_k',s_1',\dots,s_{i-k}') \circ s_{i-k+1}' \circ \dots
\circ s_n'
\)
with $u_j' \fromTB{*}{\UU} u_j$ for $1 \leqslant j \leqslant k$ and
$s_j' \fromTB{*}{\UU} s_j$ for $1 \leqslant j \leqslant n$.
Because of rightmost innermost rewriting, the terms
$\seq[k]{u}, \seq{s}$ are normal forms of $\RR$. According to
Lemma~\ref{LEM:nfb} the terms $u_1', \dots, u_k', s_1', \dots, s_n'$
are normal forms of $\UNCURRY{\RR_\eta}$.
Since $i-k \leqslant \aar(\ell)$, $\RR_\eta$ contains the rule
$\ell \circ x_1 \circ \dots \circ x_{i-k} \to r \circ x_1 \circ \dots
\circ x_{i-k}$ where $\seq[i-k]{x}$ are pairwise distinct variables not
occurring in $\ell$. Therefore
$\tau = \sigma \cup \{ x_1 \mapsto s_1, \dots, x_{i-k} \mapsto s_{i-k} \}$
is a well-defined substitution. We obtain
\begin{eqnarray*}
s &\ito^*_{\UDE{\RR}}&
f_i(\UNCURRY{u_1},\dots,\UNCURRY{u_k},\UNCURRY{s_1},\dots,
\UNCURRY{s_{i-k}}) \circ s_{i-k+1}' \circ \cdots \circ s_n' \\
&\ito_{\UDE{\RR}}&
\UNCURRY{(r \circ x_1 \circ \cdots \circ x_{i-k})}\UNCURRY{\tau} \circ
s_{i-k+1}' \circ \cdots \circ s_n' \\
&\fromTB{*}{\UU}&
(r \circ x_1 \circ \cdots \circ x_{i-k})\tau \circ
s_{i-k+1} \circ \cdots \circ s_n ~=~
r\sigma \circ s_1 \circ \cdots \circ s_n ~=~ t
\end{eqnarray*}
where we use the confluence of $\UU$ in the first sequence.
\item
In the second case we have $C = s_1 \circ C'$. Clearly
$C'[\ell\sigma] \rito_\RR C'[r\sigma]$.
If $\aar(s_1) \leqslant 0$ or if $\aar(s_1)$ is undefined or if
$\aar(s_1) > 0$ and the outermost $\circ$ has not been uncurried in the
sequence from $t$ to $s$ then
\(
s = s_1' \circ s' \fromTB{*}{\UU} s_1 \circ C'[\ell\sigma] \rito_\RR
s_1 \circ C'[r\sigma] = u
\)
with $s_1' \fromTB{*}{\UU} s_1$ and $s' \fromTB{*}{\UU} C'[\ell\sigma]$.
If $\aar(s_1) > 0$ and the outermost $\circ$ has been uncurried in the
sequence from $t$ to $s$ then we may write
$s_1 = f \circ u_1 \circ \dots \circ u_k$ where $k < \aar(f)$.
We have $s = f_{k+1}(\seq[k]{u'},s')$ for some term $s'$ with
$s' \fromTB{*}{\UU} C'[\ell\sigma]$
and $u_i' \fromTB{*}{\UU} u_i$ for $1 \leqslant i \leqslant k$.
In both cases we obtain
$s' \ito^+_{\UDE{\RR}} \cdot \fromTB{*}{\UU} C'[r\sigma]$
from the induction hypothesis.
Since innermost rewriting is closed under contexts, the desired
$s \ito^+_{\UDE{\RR}} \cdot \fromTB{*}{\UU} u$
follows.
\QED
\end{itemize}
\end{proof}

By Lemma~\ref{LEM:commute} and the equivalence of rightmost innermost
and innermost termination \cite{K00} we obtain the main result of this
section.

\smallskip

\begin{theorem}
\label{THM:isound}
An ATRS $\RR$ is innermost terminating if $\UDE{\RR}$ is innermost
terminating.
\qed
\end{theorem}

\section{Derivational Complexity}
\label{DER:main}

In this section we investigate how the uncurrying transformation 
affects
derivational complexity for full and innermost rewriting.

\subsection{Full Rewriting}

It is sound to use uncurrying as a preprocessor for proofs of upper bounds
on the derivational complexity:

\begin{theorem}
\label{THM:dc}
If $\RR$ is a terminating ATRS then
$\dc{\RR}{n} \in \OO(\dc{\UDE{\RR}}{n})$.
\end{theorem}
\begin{proof}
Consider an arbitrary maximal rewrite sequence
\(
t_0 \to_\RR t_1 \to_\RR t_2 \to_\RR \cdots \to_\RR t_m
\)
which we can transform into the sequence
\[
\UNCURRY{t_0}
\to^+_{\UDE{\RR}} \UNCURRY{t_1}
\to^+_{\UDE{\RR}} \UNCURRY{t_2}
\to^+_{\UDE{\RR}} \cdots
\to^+_{\UDE{\RR}} \UNCURRY{t_m}
\]
using Lemma~\ref{uncurried step}.
Moreover, $t_0 \to_{\UDE{\RR}}^* \UNCURRY{t_0}$ holds. Therefore, 
$\dl{\to_\RR}{t_0} \leqslant \dl{\to_{\UDE{\RR}}}{t_0}$.
Hence $\dc{\RR}{n} \leqslant \dc{\UDE{\RR}}{n}$ holds for all 
$n \in \NN$.
\end{proof}
 
Next we show that uncurrying preserves polynomial complexity.
Hence we
disregard duplicating (exponential complexity, cf.~\cite{HM08}) and empty
(constant complexity) ATRSs.
A TRS $\RR$ is called \emph{length-reducing} if $\RR$ is non-duplicating 
and $|\ell| > |r|$ for all rules $\ell \to r \in \RR$. The following lemma
is an easy consequence of \cite[Theorem 23]{HM08}. 
Here for a relative TRS ${\RR/\SS}$ the derivational complexity
$\dc{\RR/\SS}{n}$ is based on the rewrite relation $\to_{\RR/\SS}$
which is defined as $\to_\SS^* \cdot \to_\RR^{} \cdot \to_\SS^*$.

\begin{lemma}
\label{LEM:relative}
Let $\RR$ be a non-empty non-duplicating TRS over a signature containing
at least one symbol of arity at least two and let $\SS$ be a
length-reducing TRS. If $\RR \cup \SS$ is terminating then
$\dc{\RR \cup \SS}{n} \in \OO(\dc{\RR/\SS}{n})$.
\qed
\end{lemma}

Note that the above lemma does not hold if the TRS $\RR$ is empty.

\begin{theorem}
\label{dc}
Let $\RR$ be a non-empty ATRS.
If $\dc{\RR}{n}$ is in $\OO(n^k)$ then $\dc{\UNCURRY{\RR_\eta}/\UU}{n}$
and $\dc{\UDE{\RR}}{n}$ are in $\OO(n^k)$.
\end{theorem}
\begin{proof}
Let $\dc{\RR}{n}$ be in $\OO(n^k)$ and consider a maximal rewrite sequence
of $\to_{\UNCURRY{\RR_\eta}/\UU}$ starting from an arbitrary term $t_0$:
\[
t_0 \to_{\UNCURRY{\RR_\eta}/\UU} 
t_1 \to_{\UNCURRY{\RR_\eta}/\UU} \cdots \to_{\UNCURRY{\RR_\eta}/\UU} t_m
\]
By Lemma~\ref{curried step}
we obtain the sequence
$\CURRYp{t_0} \to_\RR \CURRYp{t_1} \to_\RR \cdots \to_\RR \CURRYp{t_m}$.
Thus, $\dl{\to_{\UNCURRY{\RR_\eta}/\UU}}{t_0} \leqslant 
\dl{\to_\RR}{\CURRYp{t_0}}$. Because $|\CURRYp{t_0}| \leqslant 2 |t_0|$,
we obtain $\dc{\UNCURRY{\RR_\eta}/\UU}{n} \leqslant \dc{\RR}{2n}$.
From the assumption the right-hand side is in $\OO(n^k)$, hence
$\dc{\UNCURRY{\RR_\eta}/\UU}{n}$ is in $\OO(n^k)$.
Since $\dc{\RR}{n}$ is in $\OO(n^k)$, $\RR$ must be non-duplicating
and terminating. Because $\UU$ is length-reducing, Lemma~\ref{LEM:relative}
yields that $\dc{\UDE{\RR}}{n}$ also is in $\OO(n^k)$.
\end{proof}

In practice it is recommendable to investigate
$\dc{\UNCURRY{\RR_\eta}/\UU}{n}$ instead of $\dc{\UDE{\RR}}{n}$,
see \cite{ZK10}.
The next example shows that uncurrying might be useful to enable criteria
for polynomial complexity.

\begin{example}
Consider the ATRS $\RR$ consisting of the two rules
\begin{xalignat*}{2}
\m{add} \jc x \jc \m{0} &\to x & 
\m{add} \jc x \jc (\m{s} \jc y) &\to \m{s} \jc (\m{add} \jc x \jc y)
\end{xalignat*}
The system $\UDE{\RR}$ consists of the rules
\begin{xalignat*}{3}
\m{add_2}(x,\m{0}) &\to x &
&&
\m{add_2}(x,\m{s_1}(y)) &\to \m{s_1}(\m{add_2}(x,y)) \\
\m{add_1}(x) \circ y &\to \m{add_2}(x,y)&
\m{add} \circ x &\to \m{add_1}(x)&
\m{s} \circ x &\to \m{s_1}(x)
\end{xalignat*}
The 2-dimensional TMI $\MM$
\begin{xalignat*}{2}
\m{add_2}_\MM(\vec x,\vec y) &= 
\m{\circ}_\MM(\vec x,\vec y) =
\BM
1\NE 1\NR
0\NE 1\NR
\EM \vec x + 
\BM
1\NE 1\NR
0\NE 1\NR
\EM \vec y
&
\m{add_1}_\MM(\vec x) &= 
\m{s_1}_\MM(\vec x) = 
\BM
1\NE 0\NR
0\NE 1\NR
\EM \vec x + 
\BM
0\NR
1\NR
\EM \\
\m{add}_\MM &= 
\m{s}_\MM = 
\m{0}_\MM = 
\BM
0\NR
1\NR
\EM
\end{xalignat*}
orients all rules in $\UDE{\RR}$ strictly, inducing a quadratic upper
bound on the derivational complexity of $\UDE{\RR}$ 
according to Theorem~\ref{THM:tmi}
and by Theorem~\ref{THM:dc} also of $\RR$. In contrast, 
the TRS~$\RR$ itself
does not admit such an interpretation of dimension 2. To see this, we
encoded the required condition as a satisfaction problem in non-linear
arithmetic over the integers.
\MINISMT~\cite{ZM10}%
\footnote{\url{http://cl-informatik.uibk.ac.at/software/minismt/}}
can prove this problem unsatisfiable by simplifying it into a
trivially unsatisfiable constraint. Details can be inferred from the
website mentioned in Footnote~\ref{FOO:web}.
\end{example}

\subsection{Innermost Rewriting}

Next we consider innermost derivational complexity. Let $\RR$ be an
innermost terminating TRS. From a result by
Krishna Rao~\cite[Section~5.1]{K00} which has been generalized by
van Oostrom~\cite[Theorems~2 and~3]{O07} we infer that
$\dl{\ito_\RR}{t} = \dl{\rito_\RR}{t}$ holds for all terms $t$.

\begin{theorem}
\label{THM:idc}
If $\RR$ is an innermost terminating ATRS then
$\idc{\RR}{n} \in \OO(\idc{\UDE{\RR}}{n})$.
\end{theorem}
\begin{proof}
Consider a maximal rightmost innermost rewrite sequence
$t_0 \rito_\RR t_1 \rito_\RR t_2 \rito_\RR \cdots \rito_\RR t_m$.
Using Lemma~\ref{LEM:commute} we obtain a sequence 
\[
t_0
\ito^+_{\UDE{\RR}} t_1' 
\ito^+_{\UDE{\RR}} t_2' 
\ito^+_{\UDE{\RR}} \cdots
\ito^+_{\UDE{\RR}} t_m'
\]
for terms $t_1', t_2', \dots, t_m'$ such that $t_i \to_\UU^* t_i'$ for
all  $1 \leqslant i \leqslant m$. It follows that
\(
\dl{\ito_\RR}{t_0} =
\dl{\rito_\RR}{t_0} 
\leqslant \dl{\ito_{\UDE{\RR}}}{t_0}
\)
and we conclude
$\idc{\RR}{n} \in \OO(\idc{\UDE{\RR}}{n})$.
\end{proof}

As Example~\ref{EX:reflecti} 
showed, uncurrying does not preserve
innermost termination. Similarly, it does not preserve innermost polynomial
complexity even if the original ATRS has linear innermost derivational
complexity.

\begin{example}
Consider the non-duplicating ATRS $\RR$ consisting of the two rules
\begin{xalignat*}{2}
\m{f} &\to \m{s} &
\m{f} \jc (\m{s} \jc x) &\to \m{s} \jc (\m{s} \jc (\m{f} \jc x))
\end{xalignat*}
Since the second rule is never used in innermost rewriting,
$\idc{\RR}{n} \in \OO(n)$
is easily shown by induction on $n$. We show
that the innermost derivational complexity of $\UDE{\RR}$ is at least
exponential. The TRS
$\UDE{\RR}$ consists of the rules
\begin{xalignat*}{5}
\m{f} &\to \m{s} &
\m{f}_1(x) &\to \m{s}_1(x) &
\m{f}_1(\m{s}_1(x)) &\to \m{s}_1(\m{s}_1(\m{f}_1(x))) &
\m{f} \circ x &\to \m{f}_1(x) &
\m{s} \circ x &\to \m{s}_1(x)
\end{xalignat*}
and one can verify that
$\dl{\ito_{\UDE{\RR}}}{\m{f}_1^n(\m{s}_1(x))} \geqslant 2^n$ for all
$n \geqslant 1$. Hence, $\idc{\UDE{\RR}}{n+3} \geqslant 2^n$ for all
$n \geqslant 0$.
\end{example}

\section{Experimental Results}
\label{EXP:main}

The results from this paper are implemented in the termination prover
\TTTT~\cite{KSZM09}.%
\footnote{\url{http://cl-informatik.uibk.ac.at/software/ttt2/}} 
Version 7.0.2 of the termination problem data base
(TPDB)\footnote{\url{http://termination-portal.org/wiki/TPDB}}
contains 195 ATRSs for full rewriting and 
18
ATRSs for innermost
rewriting. All tests have been performed on a single core of a server
equipped with eight dual-core AMD
Opteron\textsuperscript\textregistered\xspace processors
885 running at a clock rate of 2.6 GHz and 64 GB of main memory.

\smallskip 

Experiments\footnote{%
\label{FOO:web}%
\url{http://cl-informatik.uibk.ac.at/software/ttt2/10hor/}}
give evidence that uncurrying allows to handle significantly more systems.
For proving innermost termination we considered two popular termination
methods, namely the subterm criterion \cite{HM07} and matrix
interpretations \cite{EWZ08} of dimensions one to four. 
The implementation of the latter is based on
SAT solving (cf.~\cite{EWZ08}).
For a matrix
interpretation of dimension $d$ we used $5-d$ bits to represent
natural numbers in
matrix coefficients. An additional bit was used for intermediate results.
Both methods are integrated within the dependency pair framework using 
dependency graph reasoning and usable rules as proposed
in \cite{GTS05,GTSF06,HM05}. 

\smallskip 

Table~\ref{TAB:inn} shows the number of
systems that could be proved innermost terminating. In the table
$+$ ($-$) indicates that uncurrying has (not) been used as preprocessing
step, e.g., for the subterm criterion the number of successful proofs
increases from 42 to 55 if uncurrying is used as a preprocessing
transformation. For the setting based on matrix interpretations the gains
are even larger.
In the table, the numbers in parentheses denote the dimensions of the 
matrices.

\newlength{\dist}
\settowidth{\dist}{matrix (1)}
\newcommand{\fl}[1]{\makebox[0.5\dist][r]{#1}}
\newcommand{\fr}[1]{\makebox[0.5\dist][l]{#1}}
\newlength{\distt}
\settowidth{\distt}{subterm}
\newcommand{\fsl}[1]{\makebox[0.5\distt][r]{#1}}
\newcommand{\fsr}[1]{\makebox[0.5\distt][l]{#1}}

\begin{table}[tb]
\renewcommand{\arraystretch}{1.2}
\caption{Innermost termination for 213 ATRSs.\strut}
\label{TAB:inn}
\centering
\begin{tabular}{@{}%
r@{\,$/$\,}l@{\qquad}%sc
r@{\,$/$\,}l@{\qquad}%m1
r@{\,$/$\,}l@{\qquad}%m2
r@{\,$/$\,}l@{\qquad}%m3
r@{\,$/$\,}l%m4
@{}}
\hline
\multicolumn{2}{@{}c@{\qquad}}{subterm} &
\multicolumn{2}{c@{\qquad}}{matrix (1)} &
\multicolumn{2}{c@{\qquad}}{matrix (2)} &
\multicolumn{2}{c@{\qquad}}{matrix (3)} &
\multicolumn{2}{c@{}}{matrix (4)} 
\\
$-$&$+$ & 
\fl{$-$}&\fr{$+$} & 
\fl{$-$}&\fr{$+$} &
\fl{$-$}&\fr{$+$} & 
\fl{$-$}&\fr{$+$}
\\
 \fsl{42} & \fsr{55}  &
 \fl{67}  & \fr{102} &
 \fl{111} & \fr{142} &
 \fl{113} & \fr{144} &
 \fl{114} & \fr{145}
\\
\hline
\end{tabular}
\renewcommand{\arraystretch}{1}
\end{table}

\smallskip 

Table~\ref{TAB:dc} shows how uncurrying improves the performance of
\TTTT for derivational complexity. In this table we used 
TMIs as presented in Theorem~\ref{THM:tmi}.
Coefficients of TMIs
are represented with $\max \{ 2, 5-d \}$
bits; again an additional bit is allowed for intermediate results. 
If uncurrying is used as preprocessing transformation, 
TMIs can, e.g., show 14 systems to have at most quadratic
derivational complexity while without uncurrying the method only applies
to 10 systems. Since \TTTT has no special methods for proving
\emph{innermost} derivational complexity, 
the numbers in rows \textrm{dc} and \textrm{idc} coincide.

\settowidth{\dist}{TMI (1)}
\renewcommand{\fl}[1]{\makebox[0.5\dist][r]{#1}}
\renewcommand{\fr}[1]{\makebox[0.5\dist][l]{#1}}
\begin{table}[tb]
\renewcommand{\arraystretch}{1.2}
\caption{(Innermost) derivational complexity for 195 (213) ATRSs.\strut}
\label{TAB:dc}
\centering
\begin{tabular}{@{\ }%
 l@{\qquad}%label
 r@{\,$/$\,}l@{\qquad}%m1
 r@{\,$/$\,}l@{\qquad}%m2
 r@{\,$/$\,}l@{\qquad}%m3
 r@{\,$/$\,}l%m4
 @{}}
\hline
&
\multicolumn{2}{@{}c@{\qquad}}{TMI (1)} &
\multicolumn{2}{@{}c@{\qquad}}{TMI (2)} &
\multicolumn{2}{@{}c@{\qquad}}{TMI (3)} &
\multicolumn{2}{@{}c@{}}{TMI (4)} 
\\
& $-$&$+$ & $-$&$+$ & $-$&$+$ & $-$&$+$
\\
\textrm{dc}  & \fl{3}&\fr{4} & \fl{10}&\fr{14} & \fl{12}&\fr{26} &
\fl{12}&\fr{28} \\
\textrm{idc} & \fl{3}&\fr{4} & \fl{10}&\fr{14} & \fl{12}&\fr{26} &
\fl{12}&\fr{28} \\
\hline
\end{tabular}
\renewcommand{\arraystretch}{1}
\end{table}

\section{Conclusion}
\label{conclusion}

In this paper we studied properties of the uncurrying transformation
from~\cite{HMZ08} for innermost rewriting and (innermost) derivational
complexity. The significance of these results has been confirmed
empirically.

For proving (innermost) termination of applicative systems we mention
transformation~$\AA$~\cite{GTS05} as related work. The main benefit
of the approach in~\cite{GTS05} is that in contrast to our setting no
auxiliary
uncurrying rules are necessary. However, transformation~$\AA$ only works
for \emph{proper} ATRSs without head variables in the (left- and)
right-hand sides of rewrite rules. Here proper means that any constant
always appears with the same applicative arity. 

We are not aware of other investigations dedicated to (derivational)
complexity analysis of ATRSs.
However, we remark that transformation~$\AA$ preserves derivational 
complexity.This is straightforward from \cite[Lemma~2.1(3)]{KKSV96}.

As future work we plan to incorporate the results for innermost
termination into the dependency pair processors presented in~\cite{HMZ08}.

\end{document}